\documentclass[10pt,journal,twocolumn]{IEEEtran}
\usepackage{url}
\usepackage{graphicx}
\usepackage[caption=false,font=normalsize,labelfont=sf,textfont=sf]{subfig}
\usepackage{tikz}
\usetikzlibrary{arrows.meta, positioning}
\usepackage{amsmath,mathtools}
\usepackage{amsthm}
\usepackage{amsfonts}
\usepackage{amssymb}
\usepackage{cite} 
\usepackage{array,soul} 
\newtheorem{theorem}{Theorem}
\newtheorem{lemma}{Lemma}

\newtheorem{corollary}{Corollary}

\usepackage{xcolor,colortbl}
\usepackage{algorithm}
\usepackage{algpseudocode}
\usepackage{multirow,tabularx}
\newcommand{\EX}{\mathbb{E}}
\usepackage{array}
\newcolumntype{P}[1]{>{\centering\arraybackslash}p{#1}}
\begin{document}
\title{Maximizing Qubit Throughput under Buffer Decoherence and Variability in Generation}
\author{
    \IEEEauthorblockN{Padma Priyanka, Avhishek Chatterjee, and Sheetal Kalyani \thanks{This work was supported in parts by ANRF India through grant CRG/2023/005345.} \thanks{The authors are with the Department of Electrical Engineering, Indian Institute of Technology Madras, Chennai, India.}}
}

\maketitle
\begin{abstract}
    Quantum communication networks require transmission of high-fidelity, uncoded qubits for applications such as entanglement distribution and quantum key distribution. However, current implementations are constrained by limited buffer capacity and qubit decoherence, which degrades qubit quality while waiting in the buffer. A key challenge arises from the stochastic nature of qubit generation: there exists a random delay ($D$) between the initiation of a generation request and the availability of the qubit. This induces a fundamental trade-off—early initiation increases buffer waiting time and hence decoherence, whereas delayed initiation leads to server idling and reduced throughput.

We model this system as an admission control problem in a finite-buffer queue, where the reward associated with each job is a decreasing function of its sojourn time. We derive analytical conditions under which a simple “no-lag” policy—where a new qubit is generated immediately upon the availability of buffer space—is optimal. To address scenarios with unknown system parameters, we further develop a Bayesian learning framework that adaptively optimizes the admission policy. In addition to quantum communication systems, the proposed model is applicable to delay-sensitive IoT sensing and service systems.
\end{abstract}
\begin{IEEEkeywords}
	Quantum Communication; Qubit Decoherence; 
	\end{IEEEkeywords}

\section{Introduction}
\label{sec:intro}
Quantum communication networks are rapidly gaining prominence through practical deployments. A crucial step in many of these practical quantum communication systems is the transmission of high-quality uncoded qubits to an intended destination. For the quantum internet, it is important to generate entangled pairs rapidly and send one qubit from each pair to the intended destination \cite{nielsenchuang2000,VanMeter2014QuantumNetworking}. For quantum key distribution, transmitting specially designed (and randomly chosen) high-quality qubits at a high rate is a crucial step \cite{nielsenchuang2000,VanMeter2014QuantumNetworking}. In some of these applications, quantum error-correcting codes may improve performance. However, current technologies do not allow for sophisticated quantum encoding and decoding, as this would require a miniature fault-tolerant quantum computer, which is not realizable at present. Hence, in current systems, uncoded qubits are transmitted for entanglement sharing and key distribution \cite{VanMeter2014QuantumNetworking,Pant2019RoutingEntanglement,Dai2020QuantumQueuingDelay}.

In \cite{chatterjee2018qubits,jagannathan2019qubits,9162076}, the classical Shannon capacity of such systems with qubit decoherence in the transmission buffer was studied. Motivated by the constraints of current quantum technologies, in this paper the metric of interest is the rate of transmission of uncoded qubits at a pre-specified acceptable quality. In addition to uncoded qubits, this work considers another practical issue faced by current quantum communication systems: transmission buffers have small capacity ($\le 2$). This is due to the fact that most quantum communication systems rely on single-photon transmission, and current photonic buffers have limited capacity \cite{tang2015storage,wei2024quantum}.

We consider the case where at most one qubit can wait in the transmission buffer while the qubit generated immediately before it is being transmitted. The decision to start the generation of the next qubit is taken based on the state of the transmission buffer. This decision must balance between two undesirable cases: (i) a qubit generated too early waits for a long time in the buffer and decoheres, and (ii) a qubit generated so late that the transmitter experiences a long idle time. These two scenarios result in poor quality and low throughput, respectively. Balancing between these two cases is key to maximizing the throughput of high-quality qubits, which we refer to as the {\em reward} in the technical sections of the paper. This quantity is sometimes referred to as goodput in the literature on communication networks \cite{kurose2005computer}.

The mathematical formulation of the above qubit transmission problem, as discussed later in Sec.~\ref{sec:system}, is equivalent to a stationary policy for admitting a job into a finite queue (buffer size $2$), where the queue is served by a server with random service times, and an admitted job arrives after a random delay from the time of its admission by the policy. The objective in this scenario is to maximize the average reward across jobs, where the reward for a job is a decreasing function of its sojourn time, i.e., the total time the job spends in the system after arrival.

Interestingly, as discussed in Sec.~\ref{sec:system}, the same mathematical formulation is also useful in delay-constrained sensing and communication for memory-constrained IoT devices. It also turns out that a similar formulation closely represents the problem of serving a large number of customers with reasonable service quality in certain online and offline service sectors. In general, understanding the trade-off between the quality of qubits and their throughput, and choosing the appropriate operating point, is crucial in many quantum communication scenarios; hence, it is a topic of active research \cite{Pant2019RoutingEntanglement,Dai2020QuantumQueuingDelay,Mandalapu2024EQCcodes,Chandra2022DirectQuantumComm,Cacciapuoti2020QuantumTeleportationInternet}.
\subsection{Main Contributions}
The first contribution of this work is the formulation of the above qubit generation problem as an admission (or calling) problem in a finite-buffer queue. The goal is to maximize the average reward of the finite-buffer queueing system, where each job accrues a reward that is a general function of its sojourn time. This formulation is shown to be useful in IoT and delay-sensitive communications, as well as in certain service platforms.

Building on this formulation, we analytically characterize the region of distributions and parameters for which the simplest implementable policy is optimal: admit (or call) a job as soon as the job in service departs. Since IoT sensors have limited computational capability, and integrating complicated semiconductor electronics with optical buffers in quantum communication systems is challenging, the above simple policy is highly desirable.

Finally, we develop a Bayesian learning framework that enables policy optimization without prior knowledge of arrival or service time distributions, thereby relaxing the classical assumption of full distributional information. The Bayesian learning framework is adaptive in nature and therefore can also handle mean drift and mean shift, which are important practical issues in quantum communications \cite{lavery2016polarization,wang2025fast}.
\subsection{Organization}
The remainder of this paper is organized as follows. Section II describes the system model and preliminaries. Section III presents the analytical framework and theoretical results. Section IV reports the simulation results. Section V concludes the paper and discusses the directions for future work.
\section{System Model}
\label{sec:system}
In a traditional queuing system, a server serves jobs sequentially in their order of arrival. While the server serves a job the jobs which arrived behind it wait in a queue, which can hold a large number of jobs. 

Here, motivated by the applications discussed above, we consider the following server and queue system.  The service times of the jobs are i.i.d. $\{S_i\}$. There is an infinite reservoir of jobs from which the queue calls jobs. However, there are random i.i.d. delays $\{D_i\}$ between calling (or admitting) a job and its arrival to the queue. At most two jobs can be in the queue at any point of time, with one of them in service.

The problem in hand is to decide when the next job should be called after the waiting job goes into service. This decision should be chosen to optimize a long-term average reward of the system, which depends on the sojourn times of the jobs. For most practical queuing systems, stationary policies are sought due to their ease of implementation and good performance 
\cite{sutton1998reinforcement,srikant2014communication}.

The most natural class of policies for this setting is the choice of a lag $\Delta$, between the waiting job entering service and calling the next job into the buffer. In this work, we consider policies for which $\Delta$ is chosen deterministically.

Let $\{T_i\}$ be the respective sojourn times of jobs $i=1, 2, \ldots$, under a particular policy taken by the queue for calling jobs. Sojourn time of a job is its service time plus the time spent waiting in the queue.
For most practical queuing systems, stationary policies are sought due to their ease of implementation. 
A stationary policy leads to a stationary and ergodic queue evolution and
hence, the rate of processing of jobs is well defined and so is the stationary distribution of sojourn times. We denote them by $\lambda$ and $T$, respectively.

For job $i$, the system earns a reward $f(T_i)$, where $f$ is a non-negative and non-increasing function of the sojourn time. Let $N(t)$ be the total number of jobs served by the queue up to time $t$. The central problem of this paper is the following.  
$$ \max_{\Delta \ge 0} \lim_{t \to \infty} \frac{1}{t} \sum_{i=1}^{N(t)} f(T_i).$$

The above limit is almost surely deterministic and finite for all $\Delta \ge 0$, since the queueing system is stable. Though the above problem is quite simple to state, related problems in the simpler case with infinite buffer, are known to be intractable \cite{mandal2023towards}.
Hence, a complete analytical or algorithmic solution for obtaining the optimal $\Delta$ in the finite buffer case is intractable. 

Here, we take an important first step by analytically characterizing practically verifiable conditions under which a simple policy of no lag, i.e. $\Delta=0$, is optimal. This intuitive policy of calling the next jobs as soon as the previous job enters service is easy to implement and hence, is the natural choice in many practical systems. Our analytical result gives conditions when this intuitive policy is indeed the best policy. Later, we propose a Bayesian learning scheme that chooses $\Delta$  adaptively by learning from the past rewards.

The above limit can be rewritten as a simple expression in terms of expectation with respect to the stationary distribution of the queue.  Clearly, for any finite $\Delta$, $N(t) \to \infty$ almost surely, if $S$ and $D$ have proper distributions (i.e., no mass at infinity). Thus, the limit $\lim_{t \to \infty} \frac{1}{t} \sum_{i=1}^{N(t)} f(T_i)$ can be rewritten as $\lim_{t \to \infty} \frac{N(t)}{t} \frac{1}{N(t)}\sum_{i=1}^{N(t)} f(T_i)$. For any finite $\Delta$, since the system is stationary and ergodic 
\cite{wolff1989stochastic,sutton1998reinforcement,srikant2014communication}, $\frac{N(t)}{t} \to \lambda$ and $\frac{1}{N(t)}\sum_{i=1}^{N(t)} f(T_i) \to \EX[f(T)]$, almost surely. Here, the expectation is with respect to the stationary distribution of the sojourn time.

Thus, the above problem can be reformulated as: find the $\Delta \ge 0$ that maximizes the expected reward, $G=\lambda~\EX[f(T)]$, which is equal to the long-term average reward.

As discussed in Sec.~\ref{sec:intro}, though quantum communication is the main motivation of the above problem, it has direct applications in multiple other domains. Here, we first discuss how the above formulation is natural for the quantum communication scenarios discussed in Sec.~\ref{sec:intro}. Followed by this we discuss how the problem emerges in applications like IoT and delay sensitive communication. Finally, we discuss its application in service sectors. 
\subsubsection*{Application in Quantum Communication}  Decoherence of uncoded qubits in the transmission buffer, and how that leads to a tradeoff between high qubit transmission rate and quality of qubits have been discussed in Sec.~\ref{sec:intro}. This tradeoff naturally leads to a natural metric: throughput or rate of transmission of qubits that are above a certain quality level (a.k.a. quantum state fidelity \cite{nielsenchuang2000}). 

In the well known erasure and depolarizing noise models for qubit decoherence \cite{nielsenchuang2000}, qubits remain uncorrupted with a probability $1-p$ and with probability $p$ it degenerates into a useless quantum state (erasure or maximally mixed state). The probability of a qubit being uncorrupted, $1-p$, decreases with the time a qubit spends in contact with the environment (sojourn time in the transmission buffer or queue, in our case) and is captured by a non-negative and non-increasing function $f(\cdot)$ of its sojourn time. Thus, the expected number of uncorrupted qubits transmitted per unit is  $\lim_{t \to \infty} \frac{1}{t} \EX \sum_{i=1}^{N(t)} f(T_i)$, which, by Fubini-Tonelli theorem \cite{durrett2019probability} is $\EX \left[\lim_{t \to \infty} \frac{1}{t}  \sum_{i=1}^{N(t)} f(T_i)\right]=\lim_{t \to \infty} \frac{1}{t}  \sum_{i=1}^{N(t)} f(T_i)$. Since the last limit is almost surely deterministic for any stationary policy for generating qubits based on the buffer state.

\subsubsection*{Applications in IoT and Delay Sensitive Communications} Wireless IoTs are often used for monitoring critical processes. Sensors in IoT devices sense the environment and share the measurements as packets over a multiple access wireless channel, which is closely modeled as a server with random service times. If the sensor sends samples or measurements frequently, the receiver or the decision center will have many highly delayed measurements, resulting in delay in decision making. On the other hand, too few samples will also lead to delay in decision making. A tradeoff would be send measurements or packets as frequently as possible while ensuring a tolerable delay \cite{Yates2021AoINetworks, 
Kaul2012StatusUpdates,Bedewy2020AoIFramework}. Mathematically, if $\{T_i\}$ are delays of the samples, the goal is to maximize the number of samples delivered per unit time with low delay, i.e., max $\EX_{\{T_i,V_i\}} \frac{1}{t} \sum_{i =1}^{N(t)} 1(T_i \le V_i)$. Here, $N(t)$ is the number of samples delivered by $t$ and $V_i$ is the threshold (possibly random) by which sample or packet $i$ must be delivered. In general, $\EX_{V_i} 1(T_i \le V_i)$  would be $f(T_i)$, where $f$ depends on the distribution of $V_i$. Thus, the  whole expression is equivalent to $ \frac{N(t)} {t} \cdot \frac{1}{N(t)} \EX_{\{T_i\}} \sum_{i =1}^{N(t)} f(T_i)$. Also, since in many applications IoT sensors are frugal, the buffer for holding samples or packets before transmission  is often small ($\ll 10$) \cite{hamidouche2019efficient}.
Thus, the queueing problem discussed above is indeed the right problem to consider in this setting as well.
 \subsubsection*{Application in Customer Service} Consider a real-time scenario where a hospital or restaurant's customer waiting room has limited capacity and can accommodate only one or two customers maximum at a time. 
If customers experience longer waiting times, their satisfaction decreases, which results in  negative online feedback or review. The service portal's effective revenue for a customer is a combination of service charge and the customer's satisfaction.
 This effective revenue can be modeled as a function  $f(T)$, where $T$ represents the total time in the system. This function $f$ is naturally non-increasing in $T$ and can be of exponential or polynomial form 
\cite{alvarado2017modeling,fang2013quantifying,loewenstein1992anomalies,mazur2013adjusting}, depending on the application. The objective here would be to maximize the total expected effective revenue per unit time, which is again $\lambda~\EX[f(T)]$.
\section{Analytical Approach}
We first consider a general non-increasing and non-negative function $f(T)$. Later we consider $f(T)=\exp(-\kappa T)$ motivated by the application in quantum communications \cite{nielsenchuang2000,jagannathan2019qubits,Pant2019RoutingEntanglement}. We also consider
 $f(T)=\frac{1}{(T+1)^ \gamma}$  inspired by the IoT sensing and customer service applications \cite{Yates2021AoINetworks,alvarado2017modeling,fang2013quantifying,loewenstein1992anomalies,mazur2013adjusting}. 

The sojourn time $T$ of a typical job in the system is its waiting time $W$ plus its service time $S$. The waiting time $W$ depends on the service of the job before it, the lag $\Delta$ in calling the current job and the delay $D$ in its arrival. Here, we assume the lag is deterministic. Thus, despite the service, delay and lag being independent, the sojourn times of the jobs are correlated. Here our interest is in the quantity $\lambda \EX[f(T)]$, which is well defined and finite.

First, we understand the relation between waiting time $W$ of the current job and the service time of the previous job, which we denote by $S_{-1}$. This would allow us to evaluate both $\EX[f( T)]$ and $\lambda$ and thus obtain $G$.

Let the interarrival time between two jobs be denoted by $\text{IAT}$. Then, by elementary renewal theorem \cite{wolff1989stochastic}, the rate of arrival of jobs $\lambda=\frac{1}{\EX[\text{IAT}]}$. A job is called $\Delta$ time after the previous job has gone into service and it arrives $D$ time after calling. Thus, IAT is nothing but the sum of the waiting time of the previous job, $\Delta$ and $D$. This is shown in Fig. \ref{fig:IAT}. Thus, $\EX[\text{IAT}]=\Delta+\EX[D]+\EX[W]$. Here, the last term is the expected waiting time of the previous job, which, by stationarity, is equal to $\EX[W]$.

Clearly, $T=W+S$ and $W$ is independent of $S$. So, $\EX[f(T)] =\EX[f(W+S)]$. For obtaining  $\EX[W]$, note that the called job reaches $\Delta+D$ time after calling and the job in service (previous job) takes $S_{-1}$ time to get served. Thus, if the called job reaches before the previous job is served, it has to wait for the remaining service to be completed.  But, if it reaches after the service has been completed, it does not wait at all. Hence, $W=\max(S_{-1}-\Delta-D,0)$.
Thus, for a deterministically chosen lag, the expression for the general reward  becomes
\begin{equation}
\label{gen_reward}    
G = \frac{\EX[f(W+S)]}{\Delta+\EX[D]+\EX[W]}
\end{equation}
\begin{figure}[!t]
\centering
\resizebox{\columnwidth}{!}{%
\begin{tikzpicture}[
    >=Stealth,
    line/.style={thick},
    doublearrow/.style={<->, thick}
]

\draw[line] (0,0) -- (8,0);

\foreach \x in {0,2.5,5,8}
    \draw[line] (\x,0.4) -- (\x,-0.4);

\draw[doublearrow] (0,0.15) -- (2.5,0.15);
\draw[doublearrow] (2.5,0.15) -- (5,0.15);
\draw[doublearrow] (5,0.15) -- (8,0.15);

\node at (1.3,0.35) {\footnotesize Waiting Time};
\node at (3.75,0.35) {\footnotesize $\Delta$};
\node at (6.5,0.35) {\footnotesize Delay};

\node[align=center, text width=2.2cm, anchor=north] at (0,-0.5)
{\footnotesize $t=t_2$\\ 2$^{\mathrm{nd}}$ job arrives};

\node[align=center, text width=2.2cm, anchor=north] at (2.5,-0.5)
{\footnotesize $t=t_2+w$\\ 2$^{\mathrm{nd}}$ job goes for service};

\node[align=center, text width=2.2cm, anchor=north] at (5,-0.5)
{\footnotesize $t=t_2+w+\Delta$\\ Call 3$^{\mathrm{rd}}$ job };

\node[align=center, text width=2.2cm, anchor=north] at (8,-0.5)
{\footnotesize $t=t_2+w+\Delta+d$\\ 3$^{\mathrm{rd}}$ job arrives};

\end{tikzpicture}
}
\caption{Calculation of inter-arrival time}
\label{fig:IAT}
\end{figure}
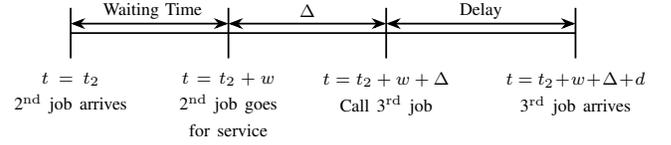
\subsection{General reward with general function f(T) for general service distribution and general delay distribution}
To determine the optimal lag, we start with the general service distribution and the general delay distribution with well-defined exponential moments for general function $f(T)$.

\begin{theorem}
\label{thm1:generalReward}
    The general reward function $G$ associated with the proposed queueing system, characterized by general service and a general delay distribution, with a general function $f$ and a deterministically chosen lag, attains its global maximum at $\Delta = 0$, if the following conditions holds:
    \begin{multline*}
     \mathbb{E}[D + S +1]\frac{\sqrt{\mathbb{E}[(f'(S))^2]}}{\mathbb{E}[f(S+S_{-1})]} \\
     \leq \frac{1}{\sqrt{\mathbb{P}(S_{-1} -D  > 0)}} - \sqrt{\mathbb{P}(S_{-1} -D  > 0)}        
    \end{multline*}
\end{theorem}
\begin{proof}
The reward function is given by
\begin{equation}
    G= \frac{\mathbb{E}[f(S+W)] }{\Delta+ \mathbb{E}[D]+\mathbb{E}[W]}    
\end{equation}
An important intermediate result that we use in this setting is the following lemma, which is proved later. 
\begin{equation*}
    \frac{d\mathbb{E}[W]}{d{\Delta}} = - \mathbb{P}(S-D> \Delta)
    \end{equation*}
To maximize the reward function over $\Delta$, the derivative of reward function w.r.t $\Delta$ is given by

\begin{equation} \label{derivativeg}
\begin{split}
\frac{dG}{d\Delta} = & \frac{\mathbb{E}\left[f'(S+W)\frac{dW}{d\Delta}\right] (\Delta+\mathbb{E}[D]+\mathbb{E}[W])}{(\Delta+\mathbb{E}[D]+\mathbb{E}[W])^2} \\
& - \frac{\left(1+\frac{d \mathbb{E}[W]}{d\Delta}\right) \mathbb{E}[f(S+W)]}{(\Delta+\mathbb{E}[D]+\mathbb{E}[W])^2}
\end{split}
\end{equation}

If $(\frac{dG}{d\Delta} < 0, \forall  \Delta \geq 0 ) \iff (\text{Numerator of equation (\ref{derivativeg})} < 0,  \forall \Delta \geq 0)$, then optimal $\Delta = 0$.\\
If f(x) is a decreasing function, then $f'(x) < 0$ and $|f'(x)|$  decreases as x increases. Therefore, 
$$f'(S+W)\frac{dW}{d\Delta} \leq |f'(S)||\mathbf{1}(S_{-1} -D >\Delta)  | $$
The first term in the numerator of the equation (\ref{derivativeg}) is upper bounded by 
\begin{align*}
   & (\Delta + \mathbb{E}[D] + \mathbb{E}[S])  \mathbb{E}[|f'(S)||\mathbf{1}(S_{-1} -D >\Delta)]\\
    & \leq  (\Delta + \mathbb{E}[D] + \mathbb{E}[S]) \sqrt{\mathbb{E}[(f'(S))^2]}\sqrt{\mathbb{E}[\mathbf{1}(S_{-1} -D >\Delta)}]\\
   & \hspace{4.5cm}\text{(Cauchy–Schwarz inequality)} \\
    & \leq  (\Delta + \mathbb{E}[D] + \mathbb{E}[S]) \sqrt{\mathbb{E}[(f'(S))^2]}\sqrt{\mathbb{P}(S_{-1}-D > \Delta)}\\
    & \hspace{4.5cm}\text{(Lemma \ref{lem:E[W]})} \\
    & \leq \mathbb{E}[D+S] \sqrt{\mathbb{E}[(f'(S))^2]\mathbb{P}(S_{-1}-D > \Delta)} \\
    &  \hspace{2.5cm}   + \sqrt{ \Delta^2 \mathbb{E}[(f'(S))^2]\mathbb{P}(S_{-1}-D > \Delta)}\\
    & \leq \mathbb{E}[D+S] \sqrt{\mathbb{E}[(f'(S))^2]\mathbb{P}(S_{-1}-D > 0)} \\
    &  \hspace{2.5cm}   + \sqrt{  \mathbb{E}[(f'(S))^2]\mathbb{P}(S_{-1}-D > 0)}\\
    & \leq \mathbb{E}[D+S+1] \sqrt{\mathbb{E}[(f'(S))^2]\mathbb{P}(S_{-1}-D > 0)}    
\end{align*}
For $\Delta \leq 1$, we have
 $$\Delta^2 \mathbb{P}(S_{-1}-D > \Delta) \leq \mathbb{P}(S_{-1}-D > \Delta) \leq \mathbb{P}(S_{-1}-D > 0).$$
For $\Delta>1$ we assumed $\Delta^2\mathbb{P}(S_{-1}-D > \Delta)$ is decreasing with $\Delta$.\\ 
The second term  in the numerator of the equation (\ref{derivativeg}) is 
\begin{align*}
&\hspace{-0.7cm} \mathbb{P}(S_{-1}-D \leq \Delta)\mathbb{E}[f(S+W)]\\
& \geq \mathbb{P}(S_{-1}-D \leq 0)\mathbb{E}[f(S+\max(S_{-1}-D,0))]\\
& \hspace{0.5cm} \text{(Both terms attain their minimum when $\Delta=0$)}\\
& \geq \mathbb{P}(S_{-1}-D \leq 0)\mathbb{E}[f(S+S_{-1})]
, \forall \Delta
\end{align*}
Therefore  $\Delta = 0 $ is optimal if
\begin{align*}
\mathbb{E}[D+S+1] &\sqrt{\mathbb{E}[(f'(S))^2]\mathbb{P}(S_{-1}-D > 0)}  \\
    & \hspace{0.5cm} \leq \mathbb{P}(S_{-1}-D \leq 0)\mathbb{E}[f(S+S_{-1})]\\
\implies  \mathbb{E}[D+S+1] &\frac{\sqrt{\mathbb{E}[(f'(S))^2]}}{\mathbb{E}[f(S+S_{-1})]} \\
& \hspace{-1cm}\leq \frac{1}{\sqrt{\mathbb{P}(S_{-1}-D > 0)}} - \sqrt{\mathbb{P}(S_{-1}-D > 0)}
\end{align*}
\end{proof}
\begin{lemma}
\label{lem:E[W]}
\begin{equation}
    \frac{d\mathbb{E}[W]}{d{\Delta}} = - \mathbb{P}(S-D> \Delta)
    \end{equation}
\end{lemma}
\begin{proof} 
Waiting time $W$ is given by 
\begin{equation}
    W=\max(S_{-1}-\Delta-D,0)
\end{equation}
$$ \frac{dW}{d\Delta} = -\mathbf{1}(S_{-1} -D >\Delta) + \alpha \mathbf{1}(S_{-1} -D  = \Delta)$$
where $-1\leq \alpha\leq 0$. The second term is zero if $S_{-1}$ and D are continuos random variables. 

Let $U = S_{-1} -D$, then 
\[
   \max(U - \Delta , 0) =
\begin{cases}
    U-\Delta& \text{if } U> \Delta\\
    0 & \text{elseif } U \leq \Delta
\end{cases}
\]
\[
  \frac{d}{d\Delta} \max(U - \Delta , 0) =
\begin{cases}
     -1 & \text{if } U> \Delta\\
    0 & \text{elseif } U \leq \Delta
\end{cases}
\]
 $$ \implies \frac{dW}{d\Delta} = -\mathbf{1}(S_{-1} -D >\Delta) $$
Therefore, 
\begin{align*}  
    \frac{d\mathbb{E}[W]}{d{\Delta}} &= 
\frac{d}{d\Delta}\mathbb{E}[\max(S_{-1}-\Delta - D, 0)]\\
&= \mathbb{E}[\frac{d}{d{\Delta}} \max(S_{-1}-\Delta - D,0)]\\
&= \mathbb{E}[\frac{d}{d{\Delta}} \max(U-\Delta,0)]\\
&= \int  \frac{d}{d{\Delta}} \max(U-\Delta,0) f_{U}(u)\, du \\
&= \int_{U>\Delta} -1 .f_{U}(u)\, du \\
&= -\int_{U>\Delta} f_{U}(u)\, du \\
&= - \mathbb{P}(S_{-1} - D > \Delta)\\
&= - \mathbb{P}(S-D > \Delta) \text{ (Since the system is stationary)}
\end{align*}

\end{proof}
The condition derived in Theorem \ref{thm1:generalReward} is a sufficient condition for the optimal lag to be $\Delta = 0$, and it holds for any function $f$. However, in practice, $f$ mostly have exponential or polynomial form 
\cite{chatterjee2018qubits, jagannathan2019qubits,9162076,siddhu2024unital, nielsenchuang2000}.

 Therefore, we first consider the exponential function $f$, which is practically motivated by applications in quantum communication, where decoherence effects naturally lead to exponential penalties in the sojourn time \cite{nielsenchuang2000,chatterjee2018qubits,jagannathan2019qubits}. We also study a polynomial reward function, motivated by IoT sensing and customer service applications in which performance metrics are closely related to the moments of delay \cite{Yates2021AoINetworks,alvarado2017modeling,fang2013quantifying,loewenstein1992anomalies,mazur2013adjusting}. For both exponential and polynomial reward functions, the resulting sufficient conditions admit insightful interpretations. In particular, the exponential reward function corresponds to moment generating functions of the sojourn time, while the polynomial reward function is directly related to negative moments. 

\subsection{General reward with exponential function f(T) for general service distribution and general delay distribution}
Here, we consider $f(T)=\exp(-\kappa T)$ and our interest is in the quantity $\lambda \EX[\exp(-\kappa T)]$, which is well defined and finite.
So, $\EX[\exp(-\kappa T)] =\EX[\exp(-\kappa W)] \EX[\exp(-\kappa S)]$. Henceforth, for any random variable $X$ and any real number $a$, we shall denote $\EX[\exp(a X)]$ by $M_X(a)$, whenever it is well defined. Note that $M_W(-\kappa)$ and $M_S(-\kappa)$ are well defined and finite since $W$ and $S$ are non-negative. Thus, $$G = \frac{M_W(-\kappa) M_S(-\kappa)}{\Delta+\EX[D]+\EX[W]}.$$

From the knowledge of the distribution of $S$, $M_S(-\kappa)$ can be obtained. 
For obtaining $M_W(-\kappa)$,the waiting time is $W=\max(S_{-1}-\Delta-D,0)$, which implies
$$M_W(-\kappa) \le \min ({M}_S(-\kappa)M_{\Delta}(\kappa) {M}_D(\kappa) ,1)$$ by Jensen's inequality and independence of $S_{-1}$, $\Delta$ and $D$.
 The term ${M}_S(-\kappa)$ replaces ${M}_{S_{-1}}(-\kappa)$ in the above expression since for a stationary system they are the same. Since the lag $\Delta$ is chosen deterministically, the term ${M}_{\Delta}(\kappa)=e^{\kappa \Delta}$.
This implies that the following expression is a surrogate for $G$ that upper-bounds it.

\begin{equation}\label{gen_reward1}    
G_{\text{sur}} =  \frac{{M}_S(-\kappa) \min ({M}_S(-\kappa) e^{\kappa \Delta} {M}_D(\kappa) ,1)}{\Delta+\EX[D]+\EX[W]}
\end{equation}
Later, we shall discuss the utility of this surrogate upper bound for exponential $f$.
\begin{corollary}
  The general reward function $G$ associated with the proposed queueing system characterized by general service and a general delay distribution for the exponential function $f(x)=\exp(-\kappa x)$ and deterministically chosen lag, attains its global maximum at $\Delta = 0$, if it satisfies the condition
 \begin{multline*}
    \kappa \mathbb{E}[D +S +1]\frac{ \sqrt{M_S(-2\kappa)}}{(M_S(-\kappa))^2} \\
     \leq \frac{1}{\sqrt{\mathbb{P}(S_{-1} -D  > 0)}} - \sqrt{\mathbb{P}(S_{-1} -D  > 0)}        
    \end{multline*}   
  
\end{corollary}
\begin{proof}
Let $f(x)=\exp(-\kappa x)$. Then, 
$\mathbb{E}[(f(S)] = \mathbb{E}[e^{-\kappa S}] = M_S(-\kappa)$, 
$\mathbb{E}[(f'(S))^2] = \mathbb{E}[e^{-2\kappa S}] = M_S(-2\kappa)$. 
    By substituting these in to the general sufficient condition for optimal $\Delta =0$    
    we get the following 
     \begin{multline*}
    \kappa \mathbb{E}[D +S +1]\frac{\sqrt{M_S(-2\kappa)}}{(M_S(-\kappa))^2} \\
     \leq \frac{1}{\sqrt{\mathbb{P}(S_{-1} -D  > 0)}} - \sqrt{\mathbb{P}(S_{-1} -D  > 0)}        
    \end{multline*}
\end{proof}

\begin{corollary}
  The general reward function $G$ associated with the proposed queueing system characterized by general service and general delay distribution for polynomial function $f(x)=\frac{1}{(x+1)^ \gamma}$ and deterministically chosen lag, attains its global maximum at $\Delta = 0$, if it satisfies the condition
 \begin{multline*}
    {\gamma} \mathbb{E}[D +S +1]\frac{ \sqrt{\mathbb{E}[(S+1)^{-2\gamma-2}]}}{\mathbb{E}[(S+S_{-1} +1)^{-\gamma}]} \\
     \leq \frac{1}{\sqrt{\mathbb{P}(S_{-1} -D  > 0)}} - \sqrt{\mathbb{P}(S_{-1} -D  > 0)}        
    \end{multline*} 
\end{corollary}
\begin{proof}
Let $f(x)= \frac{1}{(x+1)^\gamma}$. Then, 
$\mathbb{E}[(f(S)] = \mathbb{E}[(S+1)^{-\gamma}]$, 
$\mathbb{E}[(f'(S))^2] = \gamma^2\mathbb{E}[ (S+1)^{-2\gamma -2}]$. 
    By substituting these in to the general sufficient condition for optimal $\Delta =0$    
    we get the following 
     \begin{multline*}
    {\gamma} \mathbb{E}[D +S +1]\frac{ \sqrt{\mathbb{E}[(S+1)^{-2\gamma-2}]}}{\mathbb{E}[(S+S_{-1} +1)^{-\gamma}]} \\
     \leq \frac{1}{\sqrt{\mathbb{P}(S_{-1} -D  > 0)}} - \sqrt{\mathbb{P}(S_{-1} -D  > 0)}        
    \end{multline*}
\end{proof}

If we consider optimizing the surrogate reward function for the exponential function, which is an upper bound, we get a set of simpler conditions that are easier to check. Hence, we discuss them below.

\begin{theorem}
\label{thm2:expReward}
   The surrogate reward function $G_{\text{sur}}$ associated with the proposed queueing system characterized by general service and a general delay distribution for exponential function $f(x)=\exp(-\kappa x)$ and deterministically chosen lag, attains its global maximum at $\Delta = 0$, if one of the following conditions holds:
    \begin{enumerate}
\item  $M_S(-\kappa) M_D(\kappa) \geq 1$
\vspace{0.1cm}
\item $\frac{1}{\kappa} \ln \frac{1}{M_D(\kappa) M_S(-\kappa)} + E[D] + \left[E[W] \right]_{\Delta = 0} < \frac{1}{\kappa} \mathbb{P}(D > S)  $
    \end{enumerate}

\end{theorem}
The detailed proof of Theorem~\ref{thm2:expReward} is given in Appendix~\ref{proof_thm2}.
The above-derived results establish the conditions under which $\Delta = 0$ is the optimal choice to attain a maximum value of the reward function. The optimal value of $\Delta$ for different conditions is given in Table \ref{tablegencondition}.
However, it is not clear what fraction of distributions and their parameters satisfy these conditions. 

To gain a better understanding of this, we visualize the parameter space for the exponential function by plotting the regions where the optimality condition in Table \ref{tablegencondition} holds. The plotted regions are colour-coded to represent different ranges of parameter values and their corresponding effects on the system's behaviour. This helps to identify and interpret the boundaries of the parameter space where $\Delta = 0$ is optimal, providing a clearer insight into the practical applicability of the theoretical results.
\begin{figure*}[!t]
\centering
\subfloat[]{%
\includegraphics[width=0.48\textwidth]{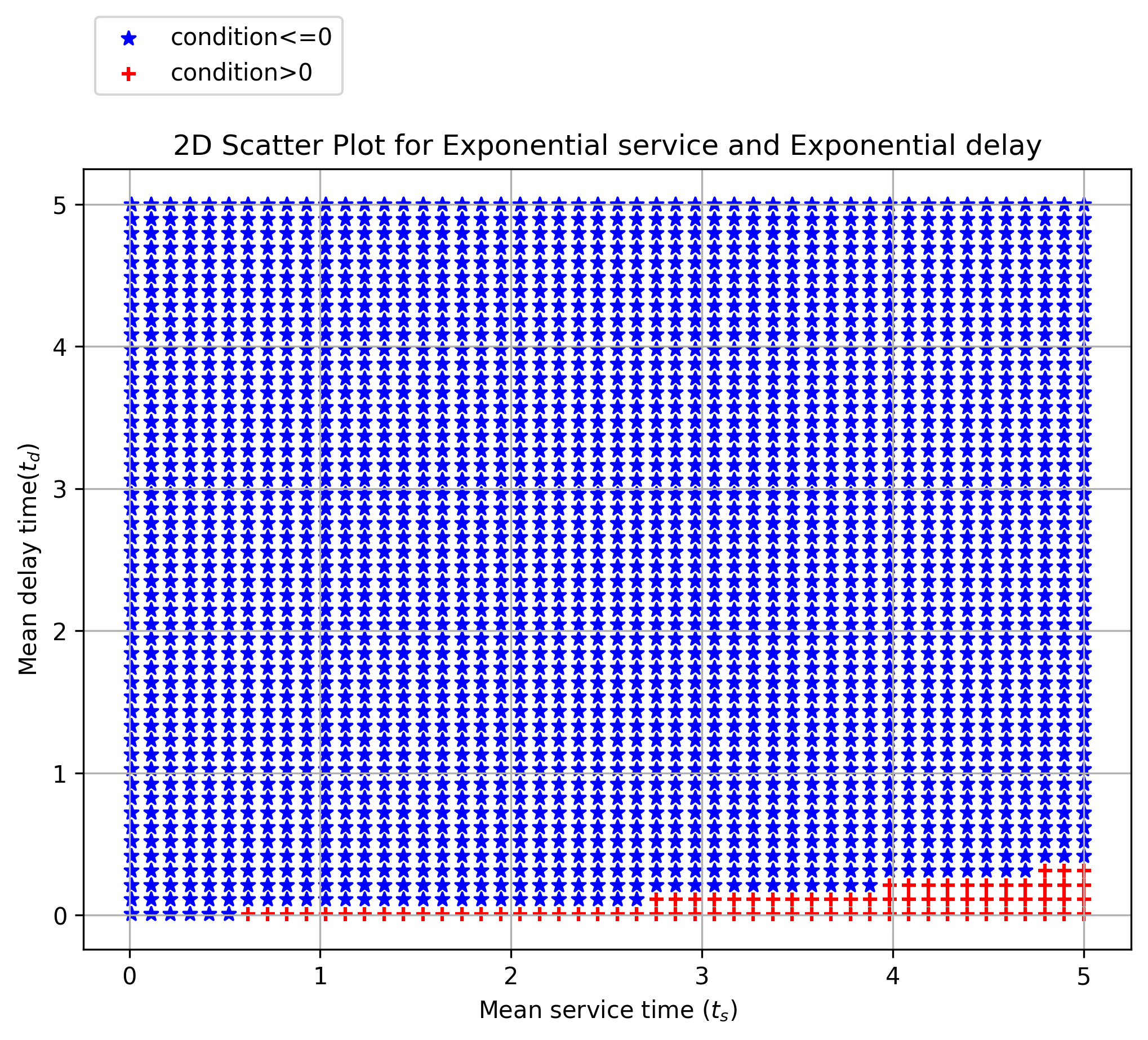}
}
\hfil
\subfloat[]{%
\includegraphics[width=0.48\textwidth]{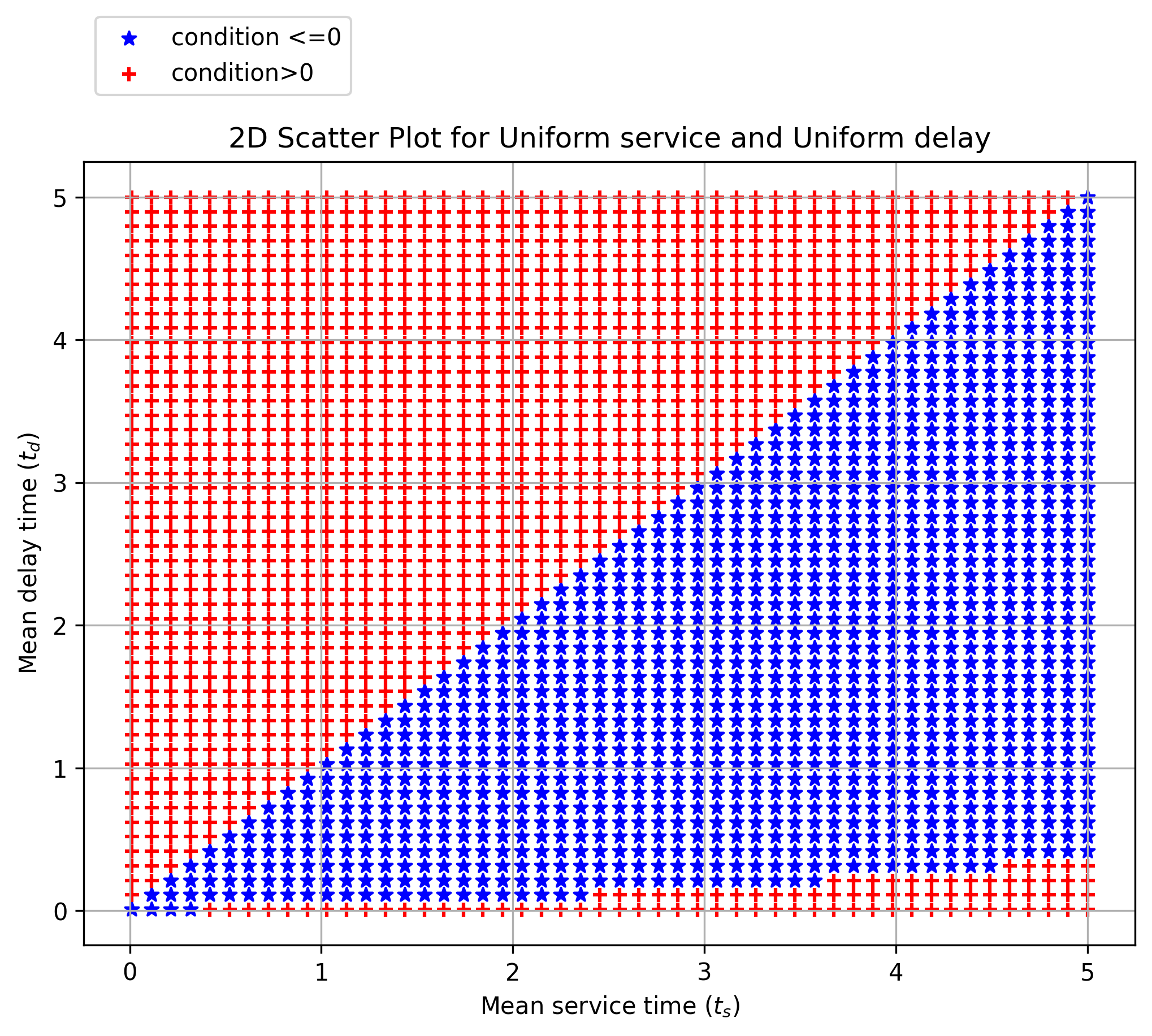}
}

\caption{2D scatter plot of the reward function versus service and delay time parameters. (a) Exponential service and exponential delay distribution. (b) Uniform service and uniform delay distribution }
\label{fig_scatter}
\end{figure*}

\par The 2D scatter plot of the reward function against service and delay time parameters is plotted in Fig. \ref{fig_scatter} for different service and delay distributions to identify the region that satisfies the desired conditions given in Table \ref{tablegencondition}. 
If the service and delay parameters are chosen within the blue region which satisfies condition 1 given in Table \ref{tablegencondition}, then $\Delta = 0$ is the optimal choice of the lag. If the parameters are chosen in the red region which does not satisfy Condition 1, the optimal lag value depends on $\kappa$, service, and delay parameters. 

\begin{table}[!t]
\caption{General distribution condition to find optimal $\Delta$}
   \label{tablegencondition}
    \centering
    \resizebox{\columnwidth}{!}{
    \renewcommand{\arraystretch}{2}
    \begin{tabular}{|P{0.1\linewidth}|P{0.6\linewidth} | P{0.3\linewidth}|} 
    \hline
   S.No & \textbf{Condition} & \textbf{Optimal $\Delta$} \\
    \hline
    \multicolumn{3}{|c|}{\textbf{$\Delta$ is chosen deterministically for any general function f }}\\
\hline
    1& $  \mathbb{E}[D + S +1]\frac{\sqrt{\mathbb{E}[(f'(S))^2]}}{\mathbb{E}[f(S+S_{-1})]}
    - \frac{1 - \mathbb{P}(S_{-1} -D  > 0)}{\sqrt{\mathbb{P}(S_{-1} -D  > 0)}} \leq 0 $ & $ \Delta =0$ \\
         \hline
      2& $  \mathbb{E}[D + S +1]\frac{\sqrt{\mathbb{E}[(f'(S))^2]}}{\mathbb{E}[f(S+S_{-1})]}
    - \frac{1 - \mathbb{P}(S_{-1} -D  > 0)}{\sqrt{\mathbb{P}(S_{-1} -D  > 0)}} > 0 $ & Depends on $\kappa$, service and delay parameter\\
    \hline
    \end{tabular}
    }

\end{table}


\subsection{Grid search as a benchmark}
A grid search approach can be used to empirically determine the optimal lag $\Delta$, even when the service and delay time distributions are unknown. 
In this approach, the queue is simulated by generating $n$ service and delay samples with the mean service time $t_s$ and the mean delay time $t_d$ from the unknown underlying service and delay distribution. For each value of $\Delta$, the reward $\frac{\EX[f(T)]}{\EX[IAT]}$ is estimated from the simulated data.
The grid search identifies the optimal calling time $\Delta_{sim}$ that has a maximum simulated reward $G_{sim}$.

The grid search can only serve as a benchmark since one will not have access to the large set of samples from the active
queue. The Bayesian learning method offers a principled alternative framework for this problem by incorporating prior information and sequentially updating posterior beliefs based on observed data.

\subsection{Bayesian Estimation of optimal $\Delta$}

In Section II-A, an optimal job-calling policy was proposed for general service and delay distributions that maximizes the reward function.
The lag is deterministic, and we showed in Theorems \ref{thm1:generalReward},\ref{thm2:expReward}  that under certain conditions, $\Delta = 0$ is optimal lag time. Here, we are trying to find the optimal $\Delta$ for cases that do not fall under these Theorems.

In the general case, mathematically, it is quite hard to get a closed-form expression for the reward function in order to find the optimal lag.
For example, when the service and delay distributions are uniform or exponential, expressions for the reward function and the corresponding conditions can be derived to determine the optimal value of $\Delta$. In contrast, obtaining such expressions becomes mathematically intractable when the distributions follow a truncated normal form. The natural question that arises is, if we cannot theoretically determine the optimal lag time, can we learn it using Bayesian methodology?
In practice, the distributions of service and delay time may
not be known. Hence, the optimal calling distribution has to
be learned over time.
\par To address this, 
we estimate the optimal lag using a parametric Bayesian learning approach. 
In this approach, the lag  $\Delta$ is modelled as a random variable with a distribution characterized by parameter  $\theta$. 
The conjugate prior distribution, namely the gamma distribution for the parameter of the estimated lag is maintained and updated for every data sample, as samples arrive.
\par The Bayesian approach relies on conjugate prior distributions to ensure analytical tractability. However, for heavy-tailed distributions, such as the Generalized Pareto Distribution (GPD) and the Hyper-Exponential Distribution (HED), conjugate priors are complex. So, we model the optimal lag distribution using an exponential distribution \cite{raj2018spectrum} since 
\begin{enumerate}
    \item It is non-negative.
    \item It is a maximum entropy distribution among all continuous PDFs defined on [0,$\Delta$]. With a fixed specified mean, it is the least informative distribution.
    \item Its conjugate prior is simple, i.e., the gamma distribution. 
\end{enumerate}

The lag is assumed to be a sample from an exponential distribution with parameter $ \theta $ i.e.,
$$ \Delta \sim  Exp(\theta) = \theta e^{-\theta t}$$
The mean of the distribution is $ 1/{\theta}$, which corresponds to the
mean lag time. 
Therefore, to determine the lag time(t), we consider the inverse of the sample obtained from the prior distribution. Physically,
$ {1}/{\theta}$ signifies the mean time duration to call the next job. As we would like to learn this value and use it as a proxy for optimal lag time.

The conjugate Gamma prior for $\theta$ is parameterized by $\alpha$ and $\beta$, i.e, $\theta \sim G(\alpha,\beta)$. 
 $${p}(\theta) = \frac{\beta^{\alpha}  }{ \Gamma(\alpha)} \theta^{\alpha-1} e^{-\beta\theta}$$
\begin{algorithm}[t] 
\caption{Proposed Algorithm}\label{algo}
\begin{algorithmic}
\State \textbf{Initialization} $\alpha_0=1,\beta_0 =1$,$\forall$ $j=1,2,...,N$.
\State  Server\_State = S, Server\_busy = $b$, Server\_{idle} = $i$.
\For{j= 1,2,...,N\text{ samples}}
\State Draw a sample $\hat{\theta} \sim G(\alpha_{j-1},\beta_{j-1})$
\State $t = \frac{1}{\hat{\theta}} $  

\If{j = 1 and $S_{j} = i$ }
 \State $\beta_j \leftarrow \beta_{j-1} + t$
     \State $\alpha_j \leftarrow \alpha_{j-1} + \epsilon_{idle}$

\ElsIf{j = 1 and $S_{j} = b$ }
 \State $\beta_j \leftarrow \beta_{j-1} + t$
     \State $\alpha_j \leftarrow \alpha_{j-1} + \epsilon_{busy}$
\Else
\If{$S_{j} = i$  and $S_{j-1} = i $  }       
            \State $\beta_j \leftarrow \beta_{j-1} + t $
            \State $\alpha_j \leftarrow \alpha_{j-1} + \epsilon_{idle}$    
\ElsIf{$S_{j} = b$  and $S_{j-1} = b$  } 
     \State $\beta_j \leftarrow \beta_{j-1} + t$
     \State $\alpha_j \leftarrow \alpha_{j-1} + \epsilon_{busy}$   
    \EndIf
    \EndIf

\EndFor
\end{algorithmic}
\end{algorithm}
\subsubsection{Working of Algorithm}
Our objective is to estimate the optimal lag time to call the next job. The queue is simulated by drawing a sample from the service and delay distribution. The lag time is found by drawing a sample 
$\hat{\theta}$ from the conjugate posterior distribution and taking the inverse.
The motivation for sampling from the Gamma posterior distribution, rather than using its mean \textbf{$\alpha/\beta$}, is to enable exploration of both smaller and larger lag values, thereby increasing the likelihood of converging to the optimal lag duration.
As the number of observations increases, the parameters of the distribution are updated, and the variance of the Gamma distribution decreases. 

The posterior distribution of $\theta$ is given by
\begin{align*}
 {p}(\theta|x) &\propto  \frac{\beta^{\alpha}  }{ \Gamma(\alpha)} \theta^{\alpha-1} e^{-\beta\theta} \theta e^{-\theta x}\\
 &\propto \theta^{\alpha+1-1}e^{-(\beta + x)\theta}\\
 {p}(\theta|x) & \sim \Gamma(\alpha +1, \beta +x)
\end{align*}

The posterior distribution seen over $n$ samples is given by
$$ {p}(\theta|x_1, x_2,...,x_n)  \sim \Gamma(\alpha +n_0, \beta +\sum_{j=1}^{n} x_j) $$ where $n_0 $ is the number of optimal lag time periods.
As the number of observed lag samples $x_j$ increases, the posterior distribution ${p}(\theta|x_1, x_2,...,x_n)$ concentrates around the maximum likelihood estimate.

\par  The hyperparameters $\alpha$ and $\beta$ are updated when both the current and previous server states are identical, i.e., both idle or both busy. If one state is idle and the other is busy, the server is considered to be in steady state, and no update is performed.  

The parameter $\alpha$  tracks the number of optimal lag periods that have been elapsed for the chosen lag sample 1/$\hat{\theta}$. 
We assume $\epsilon_{idle}$ and $\epsilon_{busy}$ optimal periods have passed if both the present and previous server status are idle and busy respectively. 
The parameters  $\epsilon_{idle}$ and $\epsilon_{busy}$
 determine how the conjugate posterior parameters are updated in response to recent server activity.
The lag time is  inversely proportional to the $\alpha$ , for each lag sample $\alpha$ is updated depending on the server status.
If the present and previous server state is idle, this indicates that a decrease in lag time is possible due to server availability which means that the mean lag time can be decreased, $\alpha$  is increased by $\epsilon_{idle}$.   
Similarly, if the present and previous server states are  busy, 
the mean lag time should be increased, $\alpha$  is increased by $\epsilon_{busy}$. This action serves to increase the lag time, likely as a response to high server load.
The parameter $\beta$ is updated by observed lag period if both present and previous server status are either idle or busy.  
 
Therefore after observing $n$ samples, $\alpha$ is updated by $n_0$, which gives number of optimal lag periods and $\beta$ gives the total duration of observed optimal lag period.

By continuously updating the parameters with each new sample, 
we arrive at an appropriate lag time distribution and the posterior mean converges to the optimal mean lag time.
Further, when the service and delay distributions themselves
change over time, and the change is unknown, the Bayesian
method can adapt to the change.

In such non-stationary settings, a Bayesian approach can adapt more naturally to these changes by continuously updating the posterior as new observations arrive, allowing it to track distributional shifts more effectively.

\section{Simulation results}
In this section, we present simulation results of the grid search and Bayesian estimation methods for estimating the optimal lag time that maximizes the reward function. In all simulations, the reward is computed using the exponential form
 $f(T)=\exp(-\kappa T)$ due to its mathematical tractability.
\subsection{{Grid search results}}
 In Fig. \ref{fig:gridsearch}, the grid search estimated reward is compared with the theoretical surrogate reward for $t_s =1$, $t_d =0.33$. We observed that the estimated reward curve for the grid search matches the theoretical surrogate reward for a chosen service and the delay time. Therefore, all the above observations reinforce our theoretical results.
\begin{figure} [t] 
    \centering  
    \includegraphics[width=0.9\columnwidth]{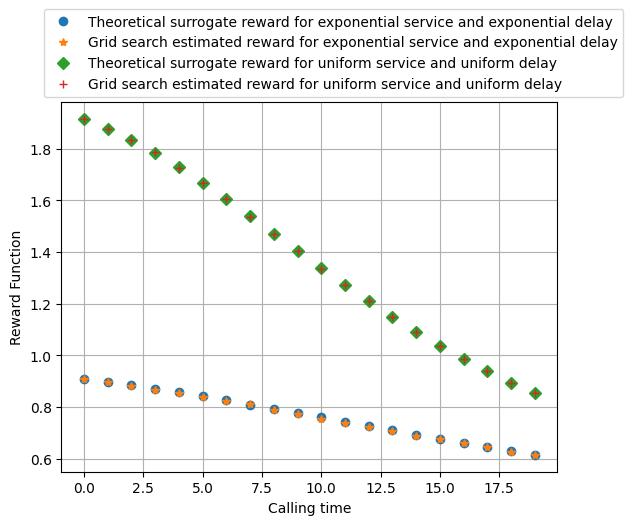}
    \caption{Comparison of Grid search estimated reward and the theoretical surrogate reward for different service and delay distributions with $t_s =1, t_d = 0.33$}
    \label{fig:gridsearch}
\end{figure}
\begin{table}[!t]
\caption{Comparison of Bayesian estimated reward and the theoretical surrogate reward for different service and delay distributions}
\label{table3}
\centering
\renewcommand{\arraystretch}{1.05}
\resizebox{\columnwidth}{!}{
\begin{tabular}{|>{\centering\arraybackslash}p{1.15cm}|>{\centering\arraybackslash}p{1.15cm}|>{\centering\arraybackslash}p{1.15cm}|>{\centering\arraybackslash}p{1.15cm}|>{\centering\arraybackslash}p{1.15cm}|}
\hline
$t_s$ &$t_d$ & $G_{sur}$ & $G_{be}$ & $G_{tb}$ \\
\hline

\multicolumn{5}{|c|}{\textbf{Case A: Exponential service and Exponential delay distribution}}\\
\hline
1&0.33&0.90910&0.8969&0.8888\\
\hline
0.5&0.1667&1.8308&1.7946&1.7779\\
\hline

\multicolumn{5}{|c|}{\textbf{Case B: Exponential service and Uniform delay distribution}}\\
\hline
1&0.33&0.9657&0.9628&0.9527\\
\hline
0.5&0.1667&1.9485&1.9359&1.9161\\
\hline

\multicolumn{5}{|c|}{\textbf{Case C: Uniform service and Uniform delay distribution}}\\
\hline
1&0.33&1.9139&1.8908&1.8934\\
\hline
0.5&0.1667&3.8410&3.7841&3.7506\\
\hline

\multicolumn{5}{|c|}{\textbf{Case D: Uniform service and Exponential delay distribution}}\\
\hline
1&0.33&1.6455&1.6451&1.6321\\
\hline
0.5&0.1667&3.2914&3.2816&3.2545\\
\hline

\multicolumn{5}{|c|}{\textbf{Case E: Exponential service and Truncated Normal delay }}\\
\multicolumn{5}{|c|}{\textbf{ distribution}}\\
\hline
1&0.33&-&0.7497&-\\
\hline
0.5&0.1667&-&1.2584&-\\
\hline

\multicolumn{5}{|c|}{\textbf{Case F: Truncated Normal service and Truncated Normal }}\\
\multicolumn{5}{|c|}{\textbf{ delay distribution}}\\
\hline
1&0.33&-&0.7419&-\\
\hline
0.5&0.1667&-&1.0148&-\\
\hline
\end{tabular}
}

\end{table}

\begin{figure} [t] 
    \includegraphics[width=0.9\columnwidth]{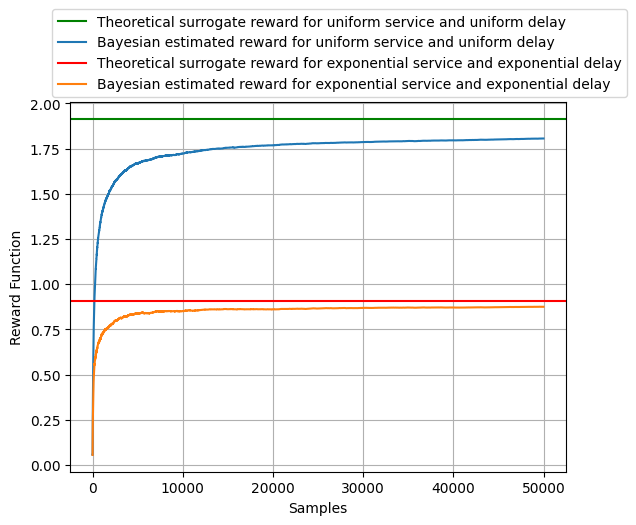}
			\caption{Comparision of Bayesian estimated reward and theoretical surrogate reward for different service and delay distributions with $t_s =1, t_d = 0.33$}		
            \label{fig_be_reward}
    
\end{figure}

\begin{figure*}[!t]
\centering

\subfloat[]{%
\includegraphics[width=0.48\textwidth]{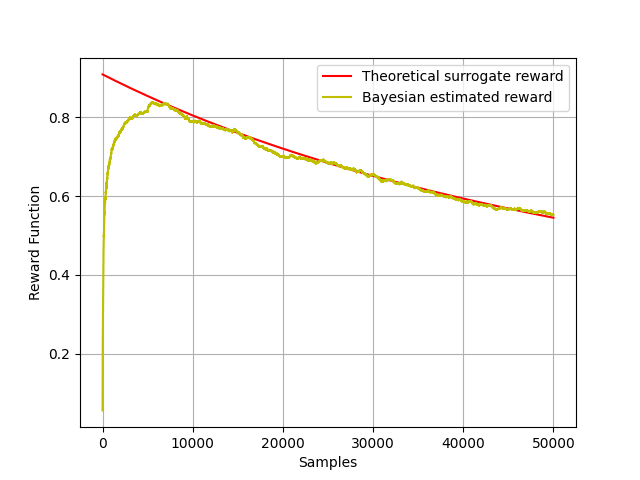}
\label{gradual}
}
\hfil
\subfloat[]{%
\includegraphics[width=0.48\textwidth]{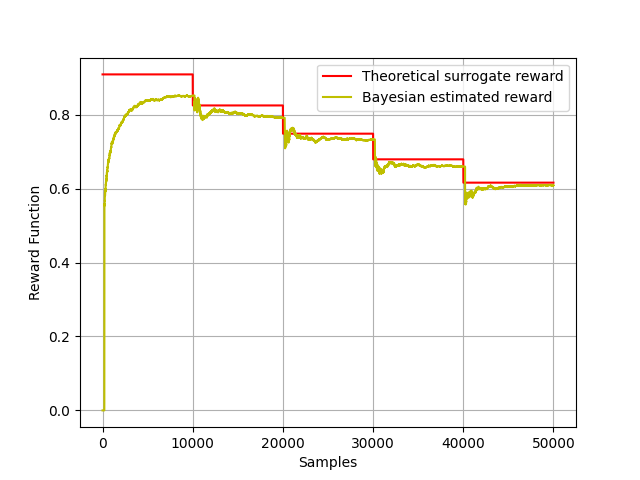}
\label{abrupt}
}

\caption{Mean-shift comparison of Bayesian estimated reward and the theoretical surrogate reward for exponential service and exponential delay distributions. (a) Gradual change in mean of service and delay distribution. (b) Abrupt change in mean of service and delay distribution. }
\label{fig_meanshift}
\end{figure*}


\subsection{Bayesian estimation results}
The optimal lag is estimated using the proposed Bayesian approach by modelling the queue with exponential or general service-time and delay-time distributions. For the general case, uniform, exponential, and truncated normal distributions are considered.
In the proposed algorithm \ref{algo}, the hyperparameters $\alpha$ and $\beta$ are adaptively updated depending on the server availability.
The parameters $\epsilon_{idle}$,$\epsilon_{busy}$ play an important role in updating the conjugate posterior distribution parameters.
To find an appropriate value for $\epsilon_{idle}$,$\epsilon_{busy}$, extensive evaluation studies are done and found that $\epsilon_{idle} = 3$,$\epsilon_{busy} = 1$ works well.
In all the cases, the queue is simulated for 50000 samples with $\epsilon_{idle} = 3$,$\epsilon_{busy} = 1$ and the optimal lag is estimated using Bayesian method.
Once the Bayesian algorithm has converged, the optimal lag is determined.
Due to inherent randomness in each run, the estimated mean lag may vary across trials. Hence, the reward function is evaluated to enable comparison with the surrogate theoretical reward. The reward function  
 $\frac{\EX[\exp(-\kappa T)]}{\EX[IAT]}$,
is calculated using last 5000 samples.

In Fig. \ref{fig_be_reward} and  Table \ref{table3}, for different service and delay distribution the Bayesian estimated reward $G_{be}$  is compared with the theoretical surrogate reward $G_{sur}$ and observed $G_{be}$ is approximately equal to $G_{sur}$.
In Fig. \ref{fig_be_reward}, it is observed that the Bayesian estimated reward rapidly converges to the theoretical surrogate reward within the initial few samples and subsequently adapts to changes in the underlying distribution.

The theoretical surrogate reward value $G_{tb}$ is calculated for the Bayesian estimated lag time 
and compared it with true theoretical surrogate reward $G_{sur}$. It was observed that $G_{tb}$ and $G_{sur}$ values were very close, confirming that Bayesian estimation is effective for finding the optimal calling time.

In Table \ref{table3}, for cases A to D, the system is stationary and we could derive the theoretical expressions for the reward function.
In Table \ref{table3}, for cases E,F the system is stationary, but obtaining a closed form theoretical expressions for the surrogate reward function is very difficult for exponential service and truncated normal delay distribution and truncated normal service and delay distributions. 
\par In these cases, we estimated the reward function $G_{be}$ using Bayesian methods and observed the reward value is close to the grid search simulated reward $G_{sim}$. 
\subsubsection{Mean shift of service and delay distribution}
We now consider scenarios in which the system's characteristics vary over time, either gradually or abruptly. 
In the proposed theory and grid search we take  expectation to calculate the reward function under the assumption of a stationary system. 
\par The proposed algorithm employs sequential Bayesian estimation to enable continuous learning, allowing it to adapt to changes in the parameters of the service and delay distributions by updating the posterior distribution, without discarding prior observations.
\par {In Fig. \ref{gradual}, the mean service and delay times are gradually changed linearly and the Bayesian estimated reward is calculated using the sliding window method.  
In Fig. \ref{abrupt}, the mean service and delay times are changed for every 10000 samples and the Bayesian estimated reward is calculated without considering the transient phase.
In Fig. \ref{fig_meanshift}, for the gradual and abrupt change in the mean service and delay distribution, the Bayesian estimated reward $G_{be}$ is approximately equal to the theoretical surrogate reward $G_{sur}$. 
In non-stationary cases, for each pair of mean service and delay times ($t_s$,$t_d$)
respectively, the theoretical reward function is evaluated over different lag-time samples, and the maximum reward is chosen as the theoretical surrogate reward.   
In contrast, the Bayesian approach does not require prior knowledge of changes in mean service and delay times, it adaptively learns and estimates the reward from the observed data.

Therefore, the proposed algorithm exhibits robustness by effectively adapting to both rapid and gradual changes in the service and delay parameters through continuous posterior updates enabled by sequential Bayesian estimation.

\section{Conclusion}

This paper addresses the practical challenge of maximizing the throughput of high-quality qubits in the presence of decoherence and limited transmission buffer capacity. By formulating the qubit generation process as an admission control problem for a finite-buffer queue, we develop a rigorous analytical framework that characterizes the trade-off between qubit fidelity and transmission rate. We derive verifiable conditions under which a simple and implementable zero-lag admission policy is optimal.

To account for the limitations of static analytical results in time-varying environments \cite{lavery2016polarization,wang2025fast}, we further propose a parametric Bayesian learning algorithm that enables the system to adaptively learn optimal generation times without prior knowledge of the underlying distributions.
\par The proposed framework provides actionable insights for the design of efficient quantum communication systems. Moreover, the results extend naturally to other delay-sensitive settings, including memory-constrained IoT networks and service systems. Future work will consider extensions to more general buffer architectures and multi-user network scenarios.

{\appendix 
\section*{Proof of Theorem 2}
\label{proof_thm2}
\begin{proof}
If  $M_S(-\kappa)M_D(\kappa) \geq 1$, the reward function is given by
\begin{equation}
    G= \frac{{M}_S(-\kappa) }{\Delta+ \mathbb{E}[D]+\mathbb{E}[W]}    
\end{equation}
To maximize the reward function over $\Delta$, the denominator should be minimized over $\Delta$. The derivative of denominator is given by $1 + \frac{d}{d\Delta}(\mathbb{E}[W]) $.
By Lemma~\ref{lem:E[W]}, the derivative of the denominator, given by $\mathbb{P}(S-D \leq \Delta)$, is an increasing function of $\Delta$. 
This implies that the reward function decreases as $\Delta$ increases, and the reward function $G$ is maximum at $\Delta =0 $.
\par If  $M_S(-\kappa)M_D(\kappa) < 1 $, the reward function is given by 
\begin{equation}\label{gen_reward_exp}
    G= \frac{{M}_S(-\kappa)  ({M}_S(-\kappa)e^{\kappa \Delta} {M}_D(\kappa) ) }{\Delta+ \mathbb{E}[D]+\mathbb{E}[W]}    
\end{equation}
Let $\Delta^*$ be the smallest $\Delta$ for which $M_D(\kappa) M_S(-\kappa)e^{\kappa \Delta} = 1$. 
For $\Delta > \Delta^*$, the denominator of the reward function increases with an increase in $\Delta$ while the numerator remains the same. So, the optimal $\Delta$ cannot be greater than $\Delta^*$. Thus, the search space is reduced to $[0,\Delta^*]$. Note that $\Delta^*$ is given by
\begin{equation}
    \frac{1}{\kappa}\ln \frac{1}{  M_S(-\kappa)M_D(\kappa)} 
\end{equation}
We consider $\frac{dG}{d\Delta}$  on $[0,\Delta^*]$. Note that if it is negative on $[0,\Delta^*]$, the optimal choice of $\Delta$ is $0$. 
After some algebraic manipulations on the derivative of $G$ with respect to $\Delta$, we observe that the condition 
$\frac{dG}{d\Delta} < 0 $ is equivalent to
$$ \Delta + E[D] + E[W] - \frac{1}{\kappa} \mathbb{P}(S-D \leq \Delta) < 0 $$
If the maximum value over the range $[0,\Delta^*]$ of the expression in the left hand side of the above condition is less than $0$, then  $\frac{dG}{d\Delta} < 0 $ on $[0,\Delta^*]$. 
Clearly, the expression in the left hand side of the above condition has a linear increasing term in $\Delta$ and other terms, $E[W]$ and $(- \mathbb{P}(S-D \leq \Delta))$, that decreases with  $\Delta$. Thus, if
\begin{equation}\label{maxdggeneral}
\Delta^* + E[D] + [E[W]]_{\Delta = 0} - \frac{1}{\kappa} \mathbb{P}(S-D \leq 0) < 0
\end{equation}
then $\frac{dG}{d\Delta} < 0 $ on $[0,\Delta^*]$.
Substituting  the value of $\Delta^*$ in equation (\ref{maxdggeneral}), we obtain the following condition: 
\begin{multline}\label{finalcond2}
 \frac{1}{\kappa} \ln (\frac{1}{M_D(\kappa) M_S(-\kappa)}) + E[D] + [E[W]]_{\Delta = 0}\\
 - \frac{1}{\kappa} \mathbb{P}(S-D \leq 0) < 0
\end{multline}

\par Hence, when the distribution satisfies the condition given in equation (\ref{finalcond2}), the reward function is maximized at $\Delta=0$. 
\end{proof}
}
%
\bibliographystyle{ieeetr}
\bibliography{bibtex}

\end{document}